\documentclass{theoretics}

\usepackage[utf8]{inputenc}
\usepackage{comment}
\usepackage{float}
\usepackage{xspace}
\usepackage{tikz}
\usepackage{mathtools}
\usepackage{bbm}

\usepackage[capitalise,noabbrev,nameinlink]{cleveref}
\usepackage{amsfonts}
\usepackage{thm-restate}

\newcommand{\E}{\mathbb{E}}

\newcommand{\poly}{\mathrm{poly}}

\newcommand{\pivot}{\textrm{PIVOT}\xspace}
\newcommand{\seqtruncatedpivot}{Sequential \textrm{Truncated-Pivot}\xspace}
\newcommand{\truncatedpivot}{\textrm{Truncated-Pivot}\xspace}
\newcommand{\eps}{\ensuremath{\varepsilon}}
\newcommand{\lp}{\textrm{LP}\xspace}
\newcommand{\ourcluster}{\ensuremath{C^{\textrm{pivot}}}\xspace}
\newcommand{\opt}{\textrm{OPT}\xspace}
\newcommand{\gstore}{\ensuremath{G_\mathrm{store}}\xspace}

\newcommand{\badedges}{\ensuremath{E^{\textrm{bad}}}\xspace}
\newcommand{\goodedges}{\ensuremath{E^{\textrm{good}}}\xspace}
\newcommand{\singletonedges}{\ensuremath{E^{\textrm{sin}}}\xspace}
\newcommand{\outdeg}{\ensuremath{\deg_{\textrm{out}}}\xspace}
\newcommand{\vsin}{V^\textrm{sin}\xspace}

\declaretheorem[style=italicized,numbered=no,name={Theorem\kern-.09em}]{theorem*}

\definecolor{green}{RGB}{152,252,152}

\SetKwComment{Comment}{$\triangleright$\ }{}
\let\oldnl\nl
\newcommand{\nonl}{\renewcommand{\nl}{\let\nl\oldnl}}

\makeatletter
\newcommand\IfRestateTF{%
  \ifx\label\thmt@gobble@label 
    \expandafter\@firstoftwo
  \else
    \expandafter\@secondoftwo
  \fi
}
\makeatother
\newcommand{\RestateRemark}{\IfRestateTF{{\normalfont\bfseries (Restated) }}{}}

\title{A \texorpdfstring{$(3+\varepsilon)$}{(3+epsilon)}-Approximate Correlation Clustering Algorithm in Dynamic Streams}

\ThCSauthor[aalto]{M\'elanie Cambus}{melanie.cambus@aalto.fi}[0000-0002-7635-3924]
\ThCSauthor[freiburg]{Fabian Kuhn}{kuhn@cs.uni-freiburg.de}[0000-0002-1025-5037]
\ThCSauthor[aalto]{Etna Lindy}{etna.lindy@aalto.fi}[]
\ThCSauthor[madras]{Shreyas Pai}{shreyas@cse.iitm.ac.in}[0000-0003-2409-7807]
\ThCSauthor[aalto]{Jara Uitto}{jara.uitto@aalto.fi}[0000-0002-5179-5056]
\ThCSaffil[aalto]{Aalto University, Finland}
\ThCSaffil[freiburg]{University of Freiburg, Germany}
\ThCSaffil[madras]{Indian Institute of Technology Madras, India}
\ThCSthanks{A preliminary version of this article appeared at SODA 24 \cite{CambusKLPU24}.}
\ThCSshortnames{M.\ Cambus, F.\ Kuhn, E.\ Lindy, S.\ Pai, J.\ Uitto}  
\ThCSshorttitle{A $(3+\varepsilon)$-Approximate Correlation Clustering Algorithm in Dynamic Streams}
\ThCSyear{2025}
\ThCSarticlenum{6}
\ThCSreceived{Feb 21, 2024}
\ThCSrevised{Nov 23, 2024}
\ThCSaccepted{Jan 12, 2025}
\ThCSpublished{Feb 28, 2025}
\ThCSdoicreatedtrue
\ThCSkeywords{semi-streaming, correlation clustering, dynamic streams, single pass}

\begin{document}

\maketitle

\begin{abstract}
  Grouping together similar elements in datasets is a common task in data mining and machine learning. In this paper, we study streaming algorithms for correlation clustering, where each pair of elements is labeled either similar or dissimilar. The task is to partition the elements and the objective is to minimize disagreements, that is, the number of dissimilar elements grouped together and similar elements that get separated.
  
  Our main contribution is a semi-streaming algorithm that achieves a $(3 + \varepsilon)$-approximation to the minimum number of disagreements using a single pass over the stream. In addition, the algorithm also works for dynamic streams.
  Our approach builds on the analysis of the PIVOT algorithm by Ailon, Charikar, and Newman [JACM'08] that obtains a $3$-approximation in the centralized setting.
  Our design allows us to sparsify the input graph by ignoring a large portion of the nodes and edges without a large extra cost as compared to the analysis of PIVOT.
  This sparsification makes our technique applicable in models such as semi-streaming, where sparse graphs can typically be handled much more efficiently.
  
  Our work improves on the approximation ratio of the recent single-pass $5$-approximation algorithm and on the number of passes of the recent $O(1/\varepsilon)$-pass $(3 + \varepsilon)$-approximation algorithm [Behnezhad, Charikar, Ma, Tan FOCS'22, SODA'23]. Our algorithm is also more robust and can be applied in dynamic streams.
  Furthermore, it is the first single pass $(3 + \varepsilon)$-approximation algorithm that uses polynomial post-processing time.
\end{abstract}



\section{Introduction}
In this paper, we consider the \emph{correlation clustering} problem introduced by~\cite{bansal2004correlation}, where the goal is to group together similar elements and separate dissimilar elements.
We model the similarity as a complete signed graph $G = (V, E^+ \cup E^-)$, where a \emph{positive} edge $\{ u, v \} = e \in E^+$ indicates that $u$ and $v$ are similar.
In case $e \in E^-$, the edge is \emph{negative} and the nodes are dissimilar. 
The goal is to minimize the \emph{disagreements}, where a disagreement is induced by grouping together dissimilar nodes or separating similar ones.
As pointed out by~\cite{chierichetti2014correlation}, it is typically the case that the set of negative edges is much larger than the set of positive edges.
Hence, in this paper, we identify the input graph with the set of positive edges, i.e., $G = (V, E^+)$ and the negative edges are defined implicitly\footnote{We  can instead identify the input graph with $G = (V, E^+)$ i.e. the set of negative edges. This does not make a difference for dynamic semi-streaming as we can modify the stream to first add edges between all vertex pairs, and then the stream of negative edges can be interpreted as edge deletions.}.
Correlation clustering is a natural abstraction for central problems in data mining and machine learning such as community and duplicate detection~\cite{Arasu2009, Chen2012}, link prediction~\cite{Yaroslavtsev2018}, and image segmentation~\cite{Kim2011}.
A key feature of correlation clustering, as opposed to, for example, the standard $k$-means clustering, is that the number of clusters is not predetermined.

As the volume of data sets is growing fast, there is an increasing demand for \emph{sublinear} solutions to  clustering problems.
Our main contribution is a novel sparsification technique, where we turn an input graph of $n$ nodes and $m$ edges into a sparse representation of $\widetilde{O}(n)$ bits\footnote{The $\widetilde{O}(f(n))$-notation hides polylogarithmic in $n$ terms.}.
We show how to find a $(3 + \eps)$-approximate clustering of the \emph{original} graph by only processing the sparsified graph.
This approach is appealing for many models of computation tailored for processing massive data sets such as semi-streaming, where the working space is much smaller than the size of the input graph.
We measure the space as \emph{words} of $O(\log n)$ bits, which is just enough to store an identifier of an edge.
We now state our main result and then, introduce the semi-streaming model and related work.

\begin{theorem*}[Main Theorem, informal version]
    There is a single-pass semi-streaming algorithm that obtains a $(3 + \eps)$-approximation to correlation clustering. The algorithm works even for dynamic streams. The approximation guarantee holds in expectation and with high probability\footnote{An event holds with high probability, w.h.p., if it holds with probability at least $1 - n^{-c}$ for a desirably large constant $c \geq 1$.}.
\end{theorem*}

\paragraph{State-of-the-Art in Semi-Streaming.}
In graph streaming, the input graph is given to the algorithm as an \emph{edge stream}~\cite{Feigenbaum2005, Feigenbaum2005a, Mut2005}.
In the semi-streaming setting, the algorithm has $\widetilde{O}(n)$ working space to store its state.
The goal is to make as few \emph{passes} over the edge-stream as possible, ideally just one. It is well known that there is a strong separation between one and two passes for problems like deterministic coloring \cite{assadi2022deterministic} and minimum cuts \cite{AssadiD2021}.
In the case of many problems, such as matching approximation or correlation clustering, simply storing the output might demand $\Omega(n)$ words.

For correlation clustering, it has already been observed in \cite{ahn2015correlation,behnezhad2023single} that by allowing exp\-onential-time computation, one can first run a streaming algorithm that computes an $\widetilde{O}(n)$-sized sketch of the graph that approximately stores the values of all cuts \cite{AGM12} and to then brute-force a solution by iterating over all possible clusterings. 
In this way, one obtains a $(1+\eps)$-approximation algorithm that uses $\widetilde{O}(n/\eps^2)$ space and a single pass even for dynamic streams. Note that since correlation clustering is APX-hard~\cite{charikar2005clustering}, unless $\mathsf{P}=\mathsf{NP}$, exponential-time computation is necessary for obtaining a $(1+\eps)$-approximation. The focus has therefore been on designing polynomial-time algorithms that achieve a constant approximation ratio.

Constant approximation ratios have been reached by using  the sparse-dense decomposition~\cite{cohen2021correlation, assadi2022} in single-pass semi-streaming.
On the downside, the approximation ratios, while being constant, are very high.
In the case of~\cite{cohen2021correlation}, they obtain an approximation ratio over $700$, the ratio of~\cite{assadi2022} is over $6400$. 
An $O(1/\eps)$-pass semi-streaming algorithm was given that obtains a $(3 + \eps)$-approximation to correlation clustering in \cite{Behnezhad2022}. A $5$-approximation was given using just a single pass~\cite{behnezhad2023single}.
Chakrabarty and Makarychev \cite{chakrabarty2023singlepass} improve the single-pass $5$-approximation algorithm of \cite{behnezhad2023single} to obtain a $(3 + \eps)$-approximation. For insertion-only streams, the algorithm of \cite{chakrabarty2023singlepass} only requires $O(n)$ words of space, whereas our algorithm in this case requires $O(n\log^2 n)$ words of space. However, a downside of the recent works by~\cite{Behnezhad2022, behnezhad2023single, chakrabarty2023singlepass} is that they do not work in dynamic streams.
Our algorithm and its analysis are more robust in the sense that they can be adapted to dynamic streams by using standard techniques.

\begin{remark}
    Subsequent to our work, the approximation ratio has been improved to $1.876$ by~\cite{DBLP:conf/stoc/Cohen-AddadLPTY24} using sublinear space in the streaming setting. 
\end{remark}

\subsection{Related Works on Correlation Clustering}
\label{sec:previouswork}

In the centralized setting, finding an optimal clustering that minimizes disagreements is known to be NP-hard~\cite{bansal2004correlation}, which motivates the study of approximation algorithms.
We note that there is another variant of the correlation clustering problem where we are interested in \emph{maximizing agreements}. An agreement corresponds to clustering together positive edges and separating negative edges. This variant is also NP-hard since the optimum solutions are the same for the maximization and the minimization problems.
However for approximate solutions, the two variants are very different.
For maximizing agreements, a trivial algorithm consisting in forming one single cluster or only single node clusters yields a $1/2$-approximation.
Furthermore, $0.7664$-approximation and $0.7666$-approximation algorithms are known, even for weighted graphs~\cite{swamy2004correlation, charikar2005clustering}.

In this paper, we focus on the minimizing disagreements problem. The first work to breach the integrality gap of $2$ for the standard LP relaxation of the problem was due to~\cite{cohen2022correlation,Cohen-AddadLLN2023}, it gives an approximation ratio of $(1.73+\eps)$ through rounding a solution to the Sherali-Adams relaxation. 
The current state-of-the-art approximation ratio is $1.437$ due to~\cite{DBLP:conf/stoc/CaoCL0NV24}, which is obtained by rounding the solution to the cluster LP. The cluster LP is exponentially-sized but it can be approximately solved in polynomial time and it has the advantage that we can do rounding without having to deal with correlated rounding errors.

The simple and well-known \pivot algorithm, yields a $3$-approximation~\cite{ailon2008aggregating} and is not based on solving an LP.
The \pivot algorithm works as follows.
\begin{itemize}
    \item In each sequential step, pick a node $u$ uniformly at random.
    \item Create a cluster $C$ that contains $u$ and all of its neighbors in the current graph.
    \item Remove $C$ from the graph and recurse on the remaining graph.
\end{itemize}
An equivalent formulation is through a \emph{randomized greedy Maximal Independent Set (MIS)}, where we pick a random permutation of the nodes and iterate over the nodes according to the permutation. In each step, the current node $v$ is selected to the MIS and its neighbors removed from the graph, unless $v$ was removed in an earlier step.
Through the randomized greedy MIS, one can obtain an $O(\log \log \Delta)$-pass algorithm for a $3$-approximation in semi-streaming~\cite{ahn2015correlation}.

Due to this connection to MIS, implementing the PIVOT algorithm in semi-streaming is also provably hard. There is an $\widetilde{\Omega}(n^2)$ space lower bound for computing an MIS in a single-pass of a stream of edges \cite{cormodeDK2019} and any semi-streaming algorithm using $O(n\cdot \poly\log(n))$ space for finding an MIS with constant probability of success requires $\Omega(\log\log n)$ passes \cite{AssadiKNS2024}. Nevertheless, variants of the PIVOT algorithm have been successfully shown to achieve good approximations. Our algorithm, and the works of \cite{Behnezhad2022,behnezhad2023single,chakrabarty2023singlepass} discussed earlier are all variants of the PIVOT algorithm.

\paragraph{Prior Work and Dynamic Streams}
We now elaborate on the details of \cite{behnezhad2023single,chakrabarty2023singlepass} and explain why it does not extend to dynamic streams.
To compute the $5$-approximation, \cite{behnezhad2023single} first picks a random permutation of the nodes and for each node maintains a pointer to the neighbor with the smallest rank in the permutation throughout the stream. A partial clustering is obtained based on these pointers and then unclustered nodes are put into singleton clusters. The algorithm has a linear space requirement for insertion-only streams.
The authors of \cite{chakrabarty2023singlepass} improve on this work by implementing a similar scheme but keep track of the $k$ smallest rank neighbors for each node. They show that this extension gives a $(3+O(1/k))$-approximation. Their approach requires $O(kn)$ words of space in insertion-only streams. However, finding the smallest rank neighbor for each node seems fundamentally challenging since computing a minimum in dynamic streams is provably hard. Our algorithm on the other hand can be implemented in dynamic streams with only a $\poly\log n$ space overhead, while giving the same $(3+\eps)$-approximation guarantee.

\subsection{A High Level Technical Overview of Our Contributions}

Our main contribution is a graph \emph{sparsification} technique inspired by the approaches that simulate the greedy MIS to approximate correlation clustering.
In the previous aforementioned works based on directly simulating the greedy MIS, the progress guarantee is given by a double exponential drop in the maximum degree or the number of nodes in the graph leading to $O(\log \log \Delta)$ and $O(\log \log n)$ pass algorithms.
Moreover, the handle used to obtain this degree drop is the following:
Consider the random permutation over the nodes.
After processing the first $t$ nodes, we can guarantee that the maximum degree is at most $O(n \log n)/t$ w.h.p. (see for example~\cite{ahn2015correlation}).
Furthermore, after the maximum degree is $d$, a prefix of length roughly $n/\sqrt{d}$ of the random permutation contains $O(n)$ edges. 
Then we can iterate over the permutation and get the guarantee that the maximum degree of the remaining graph is $O(\sqrt{d}\log n)$.

For a single pass semi-streaming algorithm, this approach seems fundamentally insufficient, as it relies on the progress related to the maximum degree of the graph. In our approach, we give a modified process that effectively gives a $3$-approximation in expectation (as in the greedy process), but only for \emph{almost} all nodes.
For this exposition, suppose that we have a $d$ regular graph for a sufficiently large $d$.
After we process the ``prefix'' containing the first $\Theta(n/d)$ nodes in the permutation, we expect that the degree of each node $u$ has dropped by a significant factor or a neighbor of $u$ has joined the MIS.

The key idea is that if a node $u$ is not part of this prefix, it is unlikely to join the MIS after this prefix is processed.
We leverage this idea as follows.
Prior to the simulation of the randomized greedy MIS, we ``set aside'' all nodes whose rank in the permutation is considerably larger than $n/d$.
The graph on the nodes with rank at most $n/d$ corresponds roughly to a set of nodes sampled with probability $1 / d$, which we show to contain $\widetilde{O}(n)$ edges. 

We process the prefix graph (i.e. the graph induced by sampled nodes) by running a greedy MIS algorithm on it.
This corresponds to running the \pivot algorithm on the prefix graph that does not contain any nodes that are set aside.
From prior work~\cite{ailon2008aggregating}, we almost immediately get that we do not lose more than a factor of $3$ from the optimum on the nodes clustered by the MIS on the sampled graph (\Cref{lemma: pivotcost}).

For the nodes that are set aside, we need more work. By carefully choosing the prefix length, we show that the degree of each node in the input graph drops by at least a factor of $(1 - \eps)$, with high probability, due to the greedy MIS (Lemma~\ref{lemma: low-deg-bound}).
We can charge each edge $e$ between a \pivot node and a node set aside to the greedy MIS analysis.
We then give a counting argument that shows that only an $\eps$ factor of the edges are between the nodes set aside (Lemma~\ref{lemma: goodedges}).
Hence, we can charge those to the \pivot analysis and pay only an additive $\eps$ factor in the approximation.

By setting the sampling probability appropriately, this line of attack works also for the non-regular case.
This idea is the basic building block for our results.
We note that in this sampling step, we add a $\log n$ and an $\eps$ term into the sampling probability in order to obtain a degree drop large enough for our approximation analysis and to make sure all guarantees hold with high probability.

\subsection{The Semi-Streaming Model.}\label{sec: model}
In the semi-streaming model, the input graph is not stored centrally, but an algorithm has access to the edges one by one in an input stream. A \emph{single-pass} semi-streaming algorithm has $\widetilde{O}(n)$ working space that it can use to store its state and is allowed to go through the stream only once. 

In the dynamic setting, the input stream consists of arbitrary edge insertions and deletions. Formally, the input
stream is a sequence $S = \langle s_1, s_2, . . .\rangle$ where $s_i = (e_i, \delta_i)$
where $e_i$ encodes an arbitrary undirected edge and $\delta_i \in \{-1, 1\}$. The multiplicity of an edge $e$ is defined as $f_e = \sum_{i:e_i=e} \delta_i$. Since the input graph is simple, we assume that $f_e \in \{0, 1\}$ throughout the stream for all $e$. 
At the end of the stream, we have $f_e = 1$ if $e$ belongs to the input graph and $0$ otherwise. 
In the insertion-only setting, the input stream consists only of edge insertions, i.e. $\delta_i = 1$ for all $i$. 

For the sake of clarity, we describe here \emph{how} the stream of edges is chosen. We assume that the graph is fixed before the algorithm executes, but the edges updates arrive in an adversarial order. The edge updates are revealed by the adversary, depending on the choices made by the algorithm so far. Although it is not explicitly stated, \cite{behnezhad2023single,chakrabarty2023singlepass} assume this setting for insertion-only streams. Our algorithm also works in this setting even for dynamic streams.

\subsection*{Organization of the Paper}
 The paper is organized as follows.
In section~\ref{sec: truncated-pivot}, we introduce the \truncatedpivot algorithm (\Cref{alg:delta-pivot}) for correlation clustering and show how it can be implemented in a single-pass in the dynamic and insertion-only semi-streaming models.
In section~\ref{sec: approximation-analysis}, we show that the \truncatedpivot algorithm returns a $(3+\eps)$-approximation of an optimum clustering.

\section{The \truncatedpivot Correlation Clustering Algorithm}\label{sec: truncated-pivot}

In this section, we give the \truncatedpivot algorithm for correlation clustering, which forms the basis for the semi-streaming implementation. 
The high-level idea of our algorithm is to compute a randomized greedy MIS with a small twist. 
Informally, we exclude nodes whose degree is likely to drop significantly before they are processed in the greedy MIS algorithm, where the MIS nodes will correspond to the \pivot nodes, or simply \emph{pivots}. 
This then allows us to effectively \emph{ignore} a large fraction of the nodes that will never be chosen as pivots.


\begin{algorithm}[tbh]
\caption{\truncatedpivot}\label{alg:delta-pivot}

    \KwIn{Graph $G = (V, E^+)$, each node $v \in V$ knows its degree $\deg(v)$ in~$G$}
     Fix a random permutation $\pi$ over the nodes.\;
     Initially, all nodes are unclustered and \emph{interesting}.\;
     A node $u$ marks itself \emph{uninteresting} if $\pi_{u} \geq \tau_{u}$ where $\tau_{u} = \frac{c}{\eps}\cdot \frac{n\log n}{\deg(u)}$ \label{line: interesting}\;
     Let \gstore be the graph induced by the interesting nodes.\label{line: gstore}\;
     Let $\mathcal{I}$ be the output of running greedy MIS on $\gstore$ with ordering $\pi$. \label{line:greedy-mis-prefix}\;
     Nodes in $\mathcal{I}$ become cluster centers (pivots). \label{line:pivot}\;
     Each node $u \in V \setminus \mathcal{I}$ joins the cluster of the smallest rank pivot neighbor $v$, if $\pi_v < \tau_u$.\label{line:cluster-pivot-nbrs}\;
     Each unclustered node forms a singleton cluster.\label{line: singleton}\;

\end{algorithm} 
In Section \ref{sec: approximation-analysis} we will prove the following theorem that gives a guarantee on the cost of the clustering returned by \Cref{alg:delta-pivot}.

\begin{theorem}[Main Theorem, formal version]\label{cor: parallelApproximation}
    For any $\eps \in (0,1/4)$, the \truncatedpivot algorithm (\Cref{alg:delta-pivot}) is a $(3 + \eps)$-approximation algorithm to the Correlation Clustering problem. 
    The approximation guarantee is in expectation. 
\end{theorem}

\subsection{Implementation in Dynamic Streams}

{
Here we describe and analyze \Cref{alg:dynamic-delta-pivot}, which implements \truncatedpivot in the dynamic semi-streaming model. We begin with the observation that in order to simulate \Cref{alg:delta-pivot}, we only need to store the edges incident to interesting nodes. This is because we run a greedy MIS on the graph induced by the interesting nodes, and in Line~\ref{line:cluster-pivot-nbrs}, we only cluster vertices that are neighbors of pivot (i.e. interesting) nodes.

According to Line~\ref{line: interesting} of \Cref{alg:delta-pivot}, a node $u$ marks itself uninteresting if $\pi_{u} \geq \tau_{u}$ where $\tau_{u} = c n\log n / \eps \deg(u)$. This is equivalent to saying that $u$ marks itself uninteresting if $\deg(u) \geq \sigma_{u}$ where $\sigma_{u} = c n\log n / \eps \pi_u$. Therefore, if the stream was insertion-only, we could only store the edges of $u$ as long as $\deg(u) < \sigma_{u}$.

The main challenge with dynamic streams is that we need to keep track of the incident edges of a node even if $\deg(u) \geq \sigma_{u}$, because its degree could go down later in the stream, and it could become interesting again. To overcome this, we will maintain a \emph{$k$-sparse recovery} data structure for the incident edges of each node, that allow us to recover the ($<\sigma_u$) incident edges of each interesting node $u$ at the end of the stream deterministically. This strategy is described more formally in \Cref{alg:dynamic-delta-pivot}. 

The following lemma describes a deterministic $k$-sparse recovery data structure, which follows from Lemma 9 in~\cite{DBLP:journals/tcs/BarkayPS15} (by substituting $n = k$, $u = n$, and $r = 1$ for our use case).

\begin{lemma}[Lemma 9, \cite{DBLP:journals/tcs/BarkayPS15}]\label{lem: k-sparse-recovery}
    There exists a deterministic data structure, $k$-sparse recovery with parameter $k$, that that maintains a sketch of stream $I$ (involving insertions and deletions of elements from $[n]$) and can recover all of $I$’s elements if $I$ contains at most $k$ distinct elements. It uses $O(k \log n)$ bits of space and can be updated in $O(\log^2 k)$ amortized operations.
\end{lemma}

Note that \Cref{lem: k-sparse-recovery} does not give any guarantees if the stream contains more than $k$ distinct elements. In this case, the output might be something completely meaningless. But in our use case, this only happens for uninteresting nodes, and we don't want to recover their incident edges anyway. Therefore, we are able to deterministically recover all the edges incident on interesting nodes at the end of the stream.


\begin{algorithm}[tbh]
   \setcounter{AlgoLine}{0}
\caption{Dynamic Semi-Streaming \truncatedpivot}\label{alg:dynamic-delta-pivot}
    \KwIn{Graph $G = (V, E)$ as a dynamic stream of edge insertions and deletions}
    Fix a random permutation $\pi$ over the nodes. \;
    Initially, all nodes $u$ are unclustered and \emph{interesting}, $\deg(u) = 0$, and $\sigma_u = \frac{c}{\eps} \cdot \frac{n\log n}{\pi_u}$. \;
    For each node $u$, we initialize a $\sigma_u$-sparse recovery data structure for the adjacency vector of $u$ (the row of the adjacency matrix of $G$ that corresponds to $u$). \;
    Upon receiving the $i^{th}$ element of the stream, $s_i = (e_i, \delta_i)$ where $e_i=\{u,v\}$, we update $\deg(u)$, $\deg(v)$, and the sparse recovery data structures associated with $u$ and $v$. \;
  \nonl  \textbf{At the end of the stream:}\;
    A node $u$ marks itself \emph{uninteresting} if $\deg(u) \geq \sigma_{u}$. \;
    We retrieve all incident edges of interesting nodes using the ${\sigma_u}$-sparse recovery structures for all $u$. \;
    Simulate Lines \ref{line: gstore} to \ref{line: singleton} of \Cref{alg:delta-pivot}.\;
\end{algorithm} 
We now prove a bound on the space requirement of \Cref{alg:dynamic-delta-pivot}. Note that for insertion-only streams we can get the same space guarantee by simply storing the ($<\sigma_u$) incident edges of all interesting nodes $u$.

\begin{lemma}\label{lem: space-req-dynamic-alg}
    \Cref{alg:dynamic-delta-pivot} requires $O(n\log^2 (n)/\eps)$ words of space. 
\end{lemma}
\begin{proof}
    For node $u$, ${\sigma_u}$-sparse recovery requires $O(\sigma_u\cdot\log n)$ bits of space, where $\sigma_u = c n\log n / \eps \pi_u$. 
    Since, each node requires one single ${\sigma_u}$-sparse recovery structure, the total amount of memory required to store and maintain all ${\sigma_u}$-sparse recovery structures throughout the algorithm is: 
    \begin{align*}
        \sum_{u\in V} \sigma_u\cdot \log n = \sum_{i=1}^n \frac{c n\log^2 n}{\eps\cdot i} = \frac{cn}{\eps}\log^2n \cdot \sum_{i=1}^n\frac{1}{i} = O(n\log^3 n/\eps)
    \end{align*}
    because the $n^{th}$ harmonic number is $H_n = O(\log n)$. 
    For each node, storing the current degree and the node identifier only requires $O(\log n)$ bits of memory, which makes the memory necessary for the algorithm $O(n\log^3 n/\eps)$ bits. 
    Since each word contains $O(\log n)$ bits, the lemma follows.
\end{proof}

\begin{theorem}
    \Cref{alg:dynamic-delta-pivot} computes a $(3+\eps)$-approximation in expectation of an optimum clustering in a single pass of the dynamic semi-streaming model, and it requires $O(n\log^2(n)/\eps)$ words of space.
\end{theorem} 
\begin{proof}
    By \Cref{lem: k-sparse-recovery}, in \Cref{alg:dynamic-delta-pivot}, all edges with an interesting endpoint can be recovered. Hence, \Cref{alg:dynamic-delta-pivot} works with the same set of edges as \Cref{alg:delta-pivot} when computing the clustering, thus implying that both algorithms return the same clustering. 
    
    Since \Cref{cor: parallelApproximation} implies that \Cref{alg:delta-pivot} outputs a clustering with expected cost that is a $(3+\eps)$-approximation, then \Cref{alg:dynamic-delta-pivot} does as well. 
    Additionally, \Cref{lem: space-req-dynamic-alg} states that \Cref{alg:dynamic-delta-pivot} requires $O(n\log^2(n)/\eps)$ words of space throughout the stream. 
\end{proof}

\begin{remark}
    We can run \Cref{alg:delta-pivot} independently $O(\log n)$ times and return the best solution to get a w.h.p.\ approximation guarantee using standard probability amplification arguments. This adds an additional $\log n$ factor to the space requirement. The lowest cost clustering can be found in the same pass by (approximately) evaluating the cost of each clustering on a cut-sparsifier \cite{ahn2009CutSparsifiers} (see Appendix A of \cite{behnezhad2023single} for more details).
\end{remark}

}

\section{Approximation Analysis of \truncatedpivot}\label{sec: approximation-analysis}

The goal of this section is to prove the approximation guarantee of the \truncatedpivot algorithm.
For a more comfortable analysis, we prove the approximation guarantee for a sequential version that produces the same output as \Cref{alg:delta-pivot} for each permutation. Following is the main result of this section.

\begin{restatable}{theorem}{maintheorem}
\label{thm: main}\RestateRemark
For any $\eps \in (0,1/4)$, \seqtruncatedpivot (\Cref{alg:delta-pivot-seq}) is a $(3+\eps)$-approximation algorithm to the correlation clustering problem.
The approximation guarantee is in expectation.
\end{restatable}

\paragraph*{A Sequential Process.}
Consider the following (sequential) algorithm and refer to \Cref{alg:delta-pivot-seq} for a pseudocode representation. 
Initially, each node is considered \emph{active}.
For each node $u$, we store the degree $\deg(u)$ of $u$ in the input graph.
We pick a random permutation $\pi$ on the nodes and in each iteration, we pick a node following the permutation.
If this node is still active, it is chosen as a pivot and we create a \emph{pivot cluster} consisting of the pivot node and its active neighbors (Line~\ref{line: seq-pivot} of \Cref{alg:delta-pivot-seq}).
The clustered nodes then become inactive and will not be chosen as pivots later.


\begin{algorithm}[tbh]
  \setcounter{AlgoLine}{0}
  \caption{\seqtruncatedpivot}\label{alg:delta-pivot-seq}
    \KwIn{Graph $G = (V, E^+)$, each node is active in the beginning. Let $\deg(u) = |N(u)|$ be the \emph{initial} degree of node $u$}
     Pick a random permutation $\pi$ over the nodes.\;
    \For{iteration $i = 1, 2, \ldots$  \Comment*[r]{Iterate over $\pi$}}{
     Let $\ell \coloneqq \frac{c}{\eps} \cdot \frac{n\log n}{i}$ \Comment*[r]{$c$ is a well-chosen constant.}
     Let $u \in V$ be the $i^{th}$ node in $\pi$. \;
     Each active node $v$ with $\deg(v) \ge \ell$ becomes inactive and creates a singleton cluster \label[line]{line: exclude} \;
     If $u$ is active, create a pivot cluster $C$ consisting of $u$ and its active neighbors. \label[line]{line: seq-pivot}\;
     Each node in $C$ becomes inactive. \;
     }
\end{algorithm}

Additionally, in iteration $i$, we check whether each active node $v$ has a degree significantly larger than $(n \log n)/i$.
If so, we expect that the previous pivot choices have removed a large fraction of the neighbors of $v$ from the graph.
In this case, $v$ becomes a singleton cluster (Line~\ref{line: exclude} in \Cref{alg:delta-pivot-seq}) and we charge the remaining edges of $v$ to the edges incident on neighbors that joined some pivot clusters in previous iterations.
Notice that the edges of $v$ that got removed before iteration $i$ can be due to a neighbor joining a pivot cluster or due to creating a singleton cluster.
As a technical challenge, we must show that most of the neighbors joined pivot clusters. 
Before the approximation analysis, we show in \Cref{lemma: seqvsparallel} that \Cref{alg:delta-pivot} and \Cref{alg:delta-pivot-seq} produce the same clustering if they sample the same random permutation.

\subsection{Equivalence with \truncatedpivot}

\begin{lemma} \label{lemma: seqvsparallel}
    Fix a (random) permutation $\pi$ over the nodes of $G = (V, E)$.
    Running the \seqtruncatedpivot (\Cref{alg:delta-pivot-seq}) with $\pi$ outputs the same clustering as running the \truncatedpivot (\Cref{alg:delta-pivot}) with $\pi$.
\end{lemma}

\begin{proof}
   
    Our goal is to show that both algorithms output the same clustering. 
    First, we show that in both cases, the singleton clusters are the same.
    Then, we show that in both cases the greedy MIS runs on the same subgraph, hence outputting the same pivot clusters. 
    
    Consider a node $u$ that is active in the beginning of iteration $i$ $(i \le \pi_u)$, and becomes a singleton cluster due to Line~\ref{line: exclude} of \Cref{alg:delta-pivot-seq}.
    By definition, $i$ is the smallest integer such that $\deg(u) \geq \frac{c}{\eps}\cdot \frac{n}{i}$ and therefore, $i = \lceil \tau_u \rceil$. Since $i \le \pi_u$, we have $\deg(u) \geq \frac{c}{\eps}\cdot \frac{n}{\pi_u}$, which corresponds to $u$ being uninteresting in \Cref{alg:delta-pivot}. 
    Since $u$ is in a singleton cluster, it did not join any pivot cluster, implying that no neighbor of $u$ was picked as a pivot before $u$ became a singleton cluster (i.e. $\forall v\in N(u), \pi_v>i$ or $v$ was clustered before iteration $\pi_v$). 
    Hence, no neighbor $v$ of $u$ s.t. $\pi_v<\lceil \tau_u \rceil$ becomes a pivot. Since $\pi_v$ is an integer, this is equivalent to saying no neighbor $v$ of $u$ s.t. $\pi_v<\tau_u$ becomes a pivot, so by Line~\ref{line:cluster-pivot-nbrs} of \Cref{alg:delta-pivot}, $u$ creates a singleton cluster in \Cref{alg:delta-pivot} as well.

    Now consider a node $u$ that creates a singleton cluster in \Cref{alg:delta-pivot}. 
    Node $u$ must have been labeled uninteresting (implying $\pi_u \geq \tau_u$), and $u$ can neither be a pivot nor have a neighboring pivot $v$ satisfying $\pi_v < \tau_u$. By definition of $\tau_u$, iteration $\lceil\tau_u\rceil$ is the smallest iteration such that $\deg(u) \geq \frac{c}{\eps}\cdot \frac{n}{\lceil\tau_u\rceil}$. 
    This implies that $u$ must be active at the beginning of iteration $\lceil\tau_u\rceil$ in \Cref{alg:delta-pivot-seq}, and forms a singleton cluster in that iteration.
    
    Since the nodes forming singleton clusters in both algorithms are the same, the subgraph induced by nodes not forming singleton clusters $G\left[V\setminus \vsin\right]$ is the same in both cases. 
    Both algorithms find a greedy MIS on $G\left[V\setminus \vsin\right]$, which implies that the pivot nodes will be the same in both cases.
    Finally, we observe that in both algorithms, a non-pivot node $u$ joins the cluster of the first neighbor $v$ s.t. $\pi_v < \tau_u$.
    Hence, the pivot clusters are the same for both \Cref{alg:delta-pivot-seq} and \Cref{alg:delta-pivot}.
\end{proof}

\subsection{Analyzing the Pivot Clusters}\label{sec: pivot clusters}
As the first step of our approximation analysis, we bound the number of disagreements caused by the pivot nodes and their respective clusters.
The analysis is an adaptation of the approach by \cite{ailon2008aggregating}, where we only focus on a subset of the nodes.

Recall the \pivot algorithm~\cite{ailon2008aggregating} that computes a greedy MIS. 
Initially, each node is considered \emph{active}.
The \pivot algorithm picks a random permutation of the nodes and iteratively considers each node in the permutation. 
For each active node $u$ (iterating over the permutation), \pivot forms a cluster with the active neighbors of $u$. 
The cluster is then deleted from the graph by marking the nodes in the new cluster \emph{inactive}. 
This is repeated until the graph is empty, i.e., all the nodes are clustered. 
The \pivot algorithm gives a solution with the expected cost being a $3$-approximation of the optimum solution.

The $3$-approximation given by the \pivot algorithm is due to the nature of the mistakes that can be made through the clustering process.
Consider $u, v, w \in V$: if $e_1 \coloneqq \{u, v\}$ and $e_2 \coloneqq \{v, w\}$ are in $E^+$ but $e_3 \coloneqq \{w, u\}\in E^-$, then clustering those nodes has to produce at least one mistake. 
The triplet $(e_1, e_2, e_3)$ is called a \textit{bad triangle}. 
Because a bad triangle induces at least one mistake in any clustering, even an optimum one, the number of disjoint bad triangles gives a lower bound on the disagreement produced by an optimum clustering. 
In the case of the pivot algorithm, since only direct neighbors of a pivot are added to a cluster, then the following mistakes can happen. 
Either two neighbors are included in the same cluster being dissimilar, which includes a negative edge in the cluster (the pivot was the endpoint of two positive edges in a bad triangle), or the pivot was an endpoint of the negative edge in a bad triangle which implies that only one positive edge of this bad triangle is included in the cluster and the second positive edge is cut. 
The authors of \cite{ailon2008aggregating} show that the expected number of mistakes produced by the \pivot algorithm is the sum of the probability that we make a mistake on every single bad triangle (not necessarily disjoint) in the graph. 
The $3$-approximation is obtained by comparing this expected cost to the cost of a \emph{packing \lp} which is a lower bound on the cost of an optimum clustering.
Our analysis for the mistakes caused by the pivot clusters (\Cref{lemma: expected cost,lemma: pivotcost}) is almost the same as in the previous work~\cite{ailon2008aggregating}. 
Our analysis of the singleton clusters requires us to have an explicit handle on the positive disagreements between the pivot clusters and the singleton clusters, provided by the analysis of the pivot clusters.

\paragraph{The Cost of Pivot Clusters in \seqtruncatedpivot.}
Let us phrase the expected cost of pivot clusters of \seqtruncatedpivot (Line~\ref{line: seq-pivot} of \Cref{alg:delta-pivot-seq}).
Recall that a bad triangle refers to a $3$-cycle with two positive and one negative edge.

\begin{definition} 
    Consider the set of all bad triangles $T$, and let $t \in T$ be a bad triangle on nodes $u, v$ and $w$.
    Define $A_t$ to be the event that, in some iteration, all three nodes are active and one of $\{u, v, w\}$ is chosen as a pivot (Line~\ref{line: seq-pivot} in \Cref{alg:delta-pivot-seq}). Let $p_t = \Pr[A_t]$.
\end{definition}

\begin{definition} \label{def: pivotcost}
    Let $\ourcluster$ be the cost, i.e., the number of disagreements induced by the pivot clusters (Line~\ref{line: seq-pivot} in \Cref{alg:delta-pivot-seq}). 
    For a pivot cluster $C$ created in iteration $i$, the disagreements include (1) the negative edges inside $C$ and (2) the positive edges from nodes in $C$ to nodes that are active in iteration $i$ and not contained in $C$.
    The edges that correspond to positive disagreements caused by the pivot clusters are said to be \emph{cut} by the pivot clusters.
\end{definition}

\begin{lemma}
    Let $T$ be the set of bad triangles in the input graph.
    Then, $\E[\ourcluster] \leq \sum_{t \in T} p_t$.
    \label{lemma: expected cost}
\end{lemma}

\begin{proof}
    Consider a bad triangle $t \in T$ and suppose that in some iteration $i$ all nodes in $t$ are active and one of them is chosen as a pivot node (Line~\ref{line: seq-pivot} in \Cref{alg:delta-pivot-seq}), i.e. the event $A_t$ happens at iteration $i$.
    Then, our algorithm creates one disagreement on this triangle on one of its edges $e \in t$.
    We charge this disagreement on edge $e$.
    
    We observe that each triangle $t$ can be charged at most once:
    An edge $e \in t$ is charged only if it is not incident on the pivot node and hence, cannot be charged twice in the same iteration. 
    Hence, at most one edge of $t$ can be charged in one iteration.
    Furthermore, if $e\in t$ gets charged in iteration $i$, its endpoints will not be both active in any later iteration $j > i$.
    This implies that $t$ cannot be charged again in another iteration.
    
    Also, creating clusters with neighbors can only create disagreements on bad triangles.
    Since dropping certain nodes of the graph cannot create bad triangles, the number of disagreements created on a subgraph by this process cannot be higher than the number of disagreements created on the whole graph.
    Therefore, $\E[\ourcluster] \leq \sum_{t \in T} p_t$.
\end{proof}

\paragraph*{Bounding $\opt$.}
In order to give an approximation guarantee to the clustered nodes, we first define the following fractional \lp. 
It was argued by~\cite{ailon2008aggregating} that the cost of the optimal solution $\lp_{OPT}$ to this \lp is a lower bound for the cost $\opt$ of the optimal solution for correlation clustering. Following are the primal and dual forms of this \lp, respectively:
    \begin{align}\label{formula:LP}
        \min \sum_{e\in E^{-}\cup E^{+}}x_e, \quad
        \text{s.t. } \sum_{e\in t}x_e \geq 1, \forall t\in T  \quad \qquad 
        &\max\sum_{t\in T}y_t, \quad
        \text{s.t. } \sum_{t\ni e}y_t \leq 1, \forall e\in E^-\cup E^+,
    \end{align}

where $T$ is the set of all bad triangles (non-necessarily disjoint) of the graph. 
By weak duality we have, $\sum_{t \in T}y_t \leq \lp_{OPT} \leq OPT$ for all dual feasible solutions $\{y_t\}_{t\in T}$. 
Therefore, in order to get an approximation guarantee, it suffices to compare the cost \ourcluster with a carefully constructed dual feasible solution. 
\begin{lemma}
    \label{lemma: pivotcost}
  Let \ourcluster be the number of disagreements incurred by the pivot clusters (\Cref{def: pivotcost}). We have that $\E[\ourcluster] \leq 3 \cdot \opt$.
\end{lemma}

\begin{proof}
Let $T$ be the set of bad triangles.
Recall the event $A_t$ that all nodes in $t \in T$ are active and one of the nodes in $t$ is chosen as a pivot (Line~\ref{line: seq-pivot} of \Cref{alg:delta-pivot-seq}) and let $\Pr[A_t] = p_t$.
Our goal is to use the probabilities $p_t$ to find a feasible solution to the packing \lp defined above.

Let $D_e$ be the event that \Cref{alg:delta-pivot-seq} creates a disagreement on $e$ and notice that $D_e \land A_t$ denotes the event that the disagreement caused by $A_t$ was charged on $e$.
By the definition of $A_t$, this disagreement cannot be due to creating singleton clusters in Line~\ref{line: exclude} of \Cref{alg:delta-pivot-seq}.
Consider now the event $A_t$ and observe that, as we are iterating over a random permutation of the nodes, each node in $t$ has the same probability to be chosen as the pivot (recall that the nodes of $t$ are all active by definition of $A_t$).
Furthermore, exactly one choice of pivot can cause $D_e$ for each $e \in t$.
Hence, we have that $\Pr[D_e \mid A_t] = 1/3$ and therefore, $\Pr[D_e \land A_t] = \Pr[D_e \mid A_t] \cdot \Pr[A_t] = p_t/3$.

  Consider the assignment $y_{t} = p_{t}/3$. 
  We now show that this is a feasible solution for the dual \lp in equation~(\ref{formula:LP}). This is because for all edges $e \in E^{+} \cup E^{-}$ the events $\{D_{e} \land A_{t}\}_{t \ni e}$ are disjoint from each other, and hence we have 
  \[
    \sum_{t \ni e} y_{t} = \sum_{t \ni e} \frac{p_{t}}{3} = \sum_{t \ni e}\Pr[D_{e} \land A_{t}] = \Pr[\cup_{t \ni e}D_{e} \land A_{t}] \le 1 \ .
  \]
  As this is a feasible packing, we have that $\sum_{t \ni e} p_{t}/3 \leq \opt$.
  Finally, by \Cref{lemma: expected cost} 
  \[
    \E[\ourcluster] \leq \sum_{t} p_t = 3 \cdot \sum_{t \in T} \frac{p_t}{3} \leq 3 \cdot \opt \ . \qedhere
  \]
\end{proof}

\subsection{Analyzing the Singleton Clusters}\label{sec: singleton}

The goal of this section is to bound the number of disagreements caused by the singleton clusters created in Line~\ref{line: exclude} of \Cref{alg:delta-pivot-seq}.
The high-level idea is to show that for a node $u$ of degree $\deg(u)$, either $u$ is clustered by some pivot node after $O(n / \deg(u))$ iterations or most of its edges have been cut by pivot clusters.
In the latter case, we relate the cost of the remaining edges of $u$ to the ones cut by the pivot clusters, and show that the remaining edges do not incur a large additional cost.
We also need to account for singleton clusters where most of  the edges are incident on other singleton clusters.
For this, we will do a counting argument that shows that there cannot be many singleton clusters that have many edges to other singleton clusters.

\paragraph{Charging the Edges Incident on the Singleton Clusters.}
Now, our goal is to bound the number of edges cut by the singleton clusters created in Line~\ref{line: exclude} of \Cref{alg:delta-pivot-seq}.
For intuition, consider a node $u$ and its neighbors with a \emph{smaller} degree, and suppose that $u$ will not be included in a pivot cluster.
Furthermore, suppose that roughly half of its neighbors have a smaller degree.
If any smaller degree neighbor $v$ is chosen according to the random permutation in the first (roughly) $n / \deg(u)$ iterations, then $v$ will be chosen as a pivot.
As we will show, this implies that, in expectation, almost all (roughly a $(1-\eps)$-fraction) of the smaller degree neighbors either join a pivot cluster or at least one of them will be chosen as a pivot which would include $u$ in a pivot cluster (\Cref{lemma: low-deg-bound}).
Once we have this, we can spread the disagreements on the remaining $\eps$-fraction of the edges to smaller degree nodes to the edges cut by pivot clusters.
As a technical challenge, we also need to account for nodes who have a few smaller degree neighbors to begin with.
We use a counting argument (\Cref{lemma: goodedges}) to show that a large fraction of nodes must have many neighbors in pivot clusters, which allows us to also spread the cost of the nodes with few smaller degree neighbors.

Consider a node $u$ and let $N_i(u)$ be the set of nodes at the beginning of iteration $i$ such that for each $v \in N_i(u)$, we have that $\deg(v) \leq \deg(u)$ and $v$ is not in a pivot cluster. Let $\deg_i(u) = |N_i(u)|$.
    
 \begin{lemma} \label{lemma: low-deg-bound}
     For each node $u$, at the beginning of iteration $i = \lceil\tau_u\rceil$ (recall, $\tau_u = \frac{c}{\eps} \cdot \frac{n \log n}{\deg(u)}$ from \Cref{alg:delta-pivot}), the probability that $u$ is active and $\deg_{i}(u) > \eps \cdot \deg(u)$ is upper bounded by $1/n^{c/2}$.
 \end{lemma}
 \begin{proof}
We define $A_k$ the events that $u$ is active at the beginning of iteration $k$, and the events $B_k \coloneqq \{\{\deg_k(u)>\eps \deg(u)\} \, \cap \, A_k\}$, $\forall\ 1 \le k \le i$. 
We want to show that $\Pr[B_i]\leq 1/n^{c/2}$. A useful property of these events is that $B_k \subseteq B_{k-1}$, $\forall\ 1 < k \le i$.

Using conditional probabilities we get that,
\begin{align*}
    \Pr[B_i] &= \Pr[B_i \cap B_{i-1}] = \Pr[B_i\,|\, B_{i-1}]\cdot \Pr[B_{i-1}] = \left(\prod_{k=2}^i\Pr[B_k\,|\,B_{k-1}] \right)\cdot \Pr[B_1].
\end{align*}

In the following, we use the fact that if two events $\mathcal{E}_1$ and $\mathcal{E}_2$ are such that $\mathcal{E}_1\subseteq \mathcal{E}_2$, then $\Pr[\mathcal{E}_1]\leq \Pr[\mathcal{E}_2]$. 
We also use the fact that, conditioning on $B_{k-1}$ implies that at the beginning of iteration $k-1$, there are at least $\eps \cdot \deg(u)$ nodes in $N_{k-1}(u)$. 
\begin{align*}
    \Pr[B_k^c \ |\  B_{k-1}] &= \Pr\Big[\left\{\deg_k(u) \leq \eps\cdot \deg(u)\right\}\, \cup \, A_k^c \ | \ B_{k-1}\Big]\\
    & \geq \Pr\Big[\text{$u$ becomes inactive during iteration $k-1$} \  | \  B_{k-1}\Big]\\
    & \geq \Pr\Big[\text{a node in $N_{k-1}(u)$ becomes a pivot during iteration $k-1$} \  | \  B_{k-1}\Big]\\
    &\geq \frac{\eps\cdot \deg(u)}{n-k+1}\geq \frac{\eps\cdot \deg(u)}{n}.
\end{align*}

Hence, $\Pr[B_k|B_{k-1}]\leq 1-\eps\deg(u)/n$. 
We finally get that 
\begin{align*}
    \Pr[B_i] &\leq \left(1-\frac{\eps\deg(u)}{n}\right)^{i-1} \leq \left(1-\frac{\eps\deg(u)}{n}\right)^{i/2} 
    \leq \exp\left(-\frac{\eps\deg(u)}{2n}\cdot {\frac{c}{\eps}\cdot \frac{n\log n}{\deg(u)}} \right)\leq \frac{1}{n^{c/2}} . \qedhere
\end{align*}%
\end{proof}

\begin{lemma} \label{lemma: singleton-cluster-deg-bound}
    In all iterations $i$, all nodes $u$ that are put into singleton clusters in iteration $i$ (Line~\ref{line: exclude} of \Cref{alg:delta-pivot-seq}) satisfy $\deg_{i}(u) \le \eps \cdot \deg(u)$ with probability $1-1/n^{\alpha}$ where $\alpha \coloneqq c/2 - 1 \gg 2$.
\end{lemma}
\begin{proof}
By \Cref{lemma: low-deg-bound} and union bound over all nodes, we can say that with probability at most $1/n^\alpha$, there exists a node $u$ such that at the beginning of iteration $i = \lceil\tau_u\rceil$, $u$ is active and $\deg_{i}(u) > \eps \cdot \deg(u)$. Therefore, with high probability, for all nodes $u$, at the beginning of iteration $i = \lceil\tau_u\rceil$, either $u$ is already inactive or $\deg_{i}(u) \le \eps \cdot \deg(u)$. This implies that, with high probability, if $u$ is put in a singleton cluster (Line~\ref{line: exclude} of \Cref{alg:delta-pivot-seq}), which can happen only in iteration $i = \lceil\tau_u\rceil$, we have $\deg_{i}(u) \le \eps \cdot \deg(u)$.
\end{proof}

\paragraph*{Good Edges and Counting.}
Consider a positive edge incident on a singleton cluster that contains a node $u$.
Suppose that the singleton cluster was created in iteration $i$.
We define an edge $e = \{u, v \}$ to be \emph{good} if the other endpoint, node $v$, was included in a pivot cluster (Line~\ref{line: seq-pivot} of \Cref{alg:delta-pivot-seq}) in some iteration $j < i$.
Otherwise, edge $e$ is \emph{bad}. 
The sets $\goodedges$ and $\badedges$ give a partition of $\singletonedges$, the set of edges incident to singleton clusters. 
Intuitively, if an edge is good, we can charge it to the set $\ourcluster$ which we know how to bound through \Cref{lemma: pivotcost}.
Furthermore, if we can show that most edges incident on singleton clusters are good, we can bound the cost of the bad edges.

 \begin{lemma}\label{lemma: goodedges}
     Conditioned on the high probability event of \Cref{lemma: singleton-cluster-deg-bound}$, |\badedges| \le 2\eps \cdot |\singletonedges|$.
 \end{lemma}
 \begin{proof}
        For every node $u$, we define $\deg_i(u) = |N_i(u)|$ where $N_i(u)$ is the set of neighbors of $u$ such that for $v \in N_i(u)$, $\deg(v) \leq \deg(u)$ and $v$ is not in a pivot cluster.
 
         For the analysis, let us consider the following orientation on the bad edges.
         Consider an iteration $i$, where a node $u$ is put into a singleton cluster in Line~\ref{line: exclude} of \Cref{alg:delta-pivot-seq}.
         Notice that this implies that $i = \lceil\tau_u\rceil$.
         Then, we orient each unoriented edge from $u$ to $v$ for each neighbor $v$ such that $\deg(v) \leq \deg(u)$ and $v$ is not in a pivot cluster, i.e. we orient all edges between $u$ and $N_i(u)$ from $u$ to $N_i(u)$. 
         Denote the out-degree of a node $u$ by $\outdeg(u)$, and notice that $\outdeg(u) = \deg_i(u)$. 
         Notice that $\outdeg(u)$ is a random variable. 
         
         Our conditioning on the high probability event of \Cref{lemma: singleton-cluster-deg-bound} gives $\deg_{i}(u) \le \eps \cdot \deg(u)$. 
        
        Hence, the out-degree of each singleton node $u$ verifies $\outdeg(u) \le \eps \cdot \deg(u)$. 
        Also, by definition, the out-degree of each non-singleton node is $0$.

        Let $\vsin$ be the set of nodes that are put in singleton clusters in Line~\ref{line: exclude} of \Cref{alg:delta-pivot-seq}, and let $\mathbbm{1}_{u\in \vsin}$ be the corresponding indicator random variable. 
        Notice that $|\vsin| = \sum_{u\in V}\mathbbm{1}_{u\in \vsin}$, $|\badedges|$ and $|\singletonedges|$ are random variables. 
        By definition of the orientation,
        \[
            |\badedges| = \sum_{u\in V} \outdeg(u) = \sum_{u\in V} \mathbbm{1}_{u\in \vsin}\cdot \outdeg(u),  
        \]
        since $\mathbbm{1}_{u\in \vsin} = 0$ implies $\outdeg(u) = 0$. 
        By using \Cref{lemma: low-deg-bound}, we have that 
     \[
         |\badedges| = \sum_{u\in V}\mathbbm{1}_{u\in \vsin} \cdot \outdeg(u) \leq \sum_{u\in V} \mathbbm{1}_{u\in \vsin}\cdot \eps \cdot \deg(u) .
     \]       
     By using the handshake lemma, we have that
     \[
        \sum_{u\in V} \mathbbm{1}_{u\in \vsin}\cdot \eps \cdot \deg(u) \leq 2\eps \cdot |\singletonedges| \ . \qedhere
     \]
        
 \end{proof}

We now have in hand all the necessary results to be able to prove our main theorem, which was the following. 

\maintheorem*
\begin{proof}
Recall the following definitions.
\begin{itemize}
    \item We denote the cost of the pivot clusters by \ourcluster (see \Cref{def: pivotcost}). This cost also covers the cost of the positive edges between pivot clusters and singleton clusters that were cut by the pivot clusters. These edges are called \emph{good}, and the set of those edges is denoted by \goodedges.
    \item Bad edges are the positive edges incident on singleton clusters that were not cut by the pivot cluster. Either they are between singletons or the singleton was created before the pivot cluster. Denote those edges by \badedges.
    \item $\singletonedges = \goodedges \cup \badedges$
\end{itemize}

We can split the cost of \Cref{alg:delta-pivot-seq} into two parts.
By \Cref{lemma: pivotcost}, we have that $\E[\ourcluster] \leq 3 \cdot \opt$. 
Let us define $D$ to be the event that, for all iterations $i$, all nodes $u$ that are put in singleton clusters in iteration $i$ satisfy $\deg_i(u) < \eps\cdot \deg(u)$. 
By \Cref{lemma: singleton-cluster-deg-bound}, $D$ is a high probability event. 
Then, by \Cref{lemma: goodedges}, we have that, conditioning on the high probability event $D$,
\[
    |\badedges| \leq 2\eps \cdot |\singletonedges| \leq \frac{2\eps}{1-2\eps}\cdot |\goodedges| \leq 4\eps \cdot \ourcluster \ ,
\]
where the last inequality holds because $\eps < 1/4$. 
This inequality implies that $\mathbb{E}[|\badedges|\mid D] \leq 4\eps\cdot \mathbb{E}[\ourcluster\mid D]$. And therefore,
\begin{align*}
\mathbb{E}[|\badedges|] &= \mathbb{E}[|\badedges|\mid D]\Pr[D] + \mathbb{E}[|\badedges|\mid \Bar{D}]\cdot \Pr[\Bar{D}] \\
&\leq 4\eps\cdot \mathbb{E}[\ourcluster\mid D]\cdot \left(1-\frac{1}{n^c}\right) + n^2\cdot \frac{1}{n^\alpha} \\
&\leq 4\eps\cdot \mathbb{E}[\ourcluster\mid D] + \frac{1}{n^{\alpha-2}}  \ .
\end{align*}

Also, notice that,
 \begin{align*}
     \mathbb{E}\left[\ourcluster\right] &=  \mathbb{E}\left[\ourcluster\mid D \right]\cdot \Pr[D] + \mathbb{E}\left[\ourcluster\mid \Bar{D}\right]\cdot \Pr[\Bar{D}]\\
     &\geq \mathbb{E}\left[\ourcluster  \mid D \right]\cdot \Pr[D] \\
     &\geq \mathbb{E}\left[\ourcluster\mid D \right]\cdot \left( 1-\frac{1}{n^\alpha}\right)\\
     &\geq \mathbb{E}\left[\ourcluster\mid D \right] - \frac{1}{n^{\alpha-2}} \text{ , since }\mathbb{E}\left[\ourcluster\mid D \right] \leq n^2.
 \end{align*}
 
 Which implies that 
 \begin{align*}
     \mathbb{E}\left[\ourcluster\mid D \right] \leq \mathbb{E}\left[\ourcluster\right] + \frac{1}{n^{\alpha-2}}\ .
 \end{align*}

By combining the above observations, we have that the expected cost of \Cref{alg:delta-pivot-seq} is at most
\begin{align*}
    \mathbb{E}\left[\ourcluster  + |\badedges|\right] &= \mathbb{E}\left[\ourcluster\right]  + \mathbb{E}\left[|\badedges|\right]\\
    &\leq \mathbb{E}\left[\ourcluster\right] + 4\eps\cdot \mathbb{E}[\ourcluster\mid D] + \frac{1}{n^{\alpha-2}}\\
    &\leq (1 + 4\eps)\cdot \mathbb{E}\left[\ourcluster\right] + \frac{1+4\eps}{n^{\alpha-2}}\\
    &\leq (3+12\eps) \cdot \opt + \frac{1+4\eps}{n^{\alpha-2}}.
\end{align*}

We can substitute $\eps' \coloneqq 12\eps$, where $\eps$ can be arbitrarily small. 
Notice that if $\opt \geq 1$, then we have that $\mathbb{E}\left[\ourcluster  + |\badedges|\right]\leq (3+12\eps) \cdot \opt$, which gives us a $(3+\eps')$-approximation in expectation. \qedhere

\end{proof}

\begin{remark}
    If $\opt = 0$, then the expected cost of our solution is $1/\poly(n)$ according to the proof above, or equivalently, the expected cost of our solution is $0$ with high probability.
\end{remark}

\begin{proof}[Proof of Theorem \ref{cor: parallelApproximation}]
    \Cref{cor: parallelApproximation} follows from \Cref{thm: main} and \Cref{lemma: seqvsparallel}.
\end{proof}

\section*{Acknowledgements}
We would like to thank Moses Charikar, Soheil Behnezhad, Weiyun Ma, and Li-Yang Tan for pointing out an error in an earlier analysis of our correlation clustering algorithm.
We would also like to thank Vihan Shah and Sepehr Assadi for pointing out that our algorithm works even in dynamic streaming. Finally, we thank Dennis Olivetti and Alkida Balliu for fruitful discussions.

M\'elanie Cambus is supported by Research Council of Finland Grant 334238. Part of this work was done when Shreyas Pai was a postdoctoral fellow at Aalto University, supported by Research Council of Finland Grant 334238 and Helsinki Institute for Information Technology HIIT.


\printbibliography
 


\end{document}